\documentclass[10pt]{amsart}
\usepackage[utf8]{inputenc}
\usepackage{fullpage}
\usepackage{cite}
\usepackage{bbold}
\usepackage{amsmath}
\usepackage{amsthm} 
\usepackage[foot]{amsaddr}
\usepackage{tikz}
\usetikzlibrary{%
  matrix,%
  calc,%
  arrows%
}

\usepackage{cite}
\usepackage[margin=2pt,labelfont=]{subfig}
\usepackage{color}
\usepackage{graphicx}
\usepackage{url}

\newtheorem{theorem}{Theorem}
\newtheorem{lemma}[theorem]{Lemma}

\newtheorem{definition}[theorem]{Definition}

\newcommand\Icpx{\mathbb I}

\newcommand\Ucpx{\mathbb U}

\newcommand\Xcpx{\mathbb X}

\newcommand\kk{\mathbb k}

\newcommand\Ucoll{\mathcal U}

\newcommand\img{\mathop{\mathrm{im}}}
\newcommand\coker{\mathop{\mathrm{coker}}}
\newcommand\Tot{\mathop{\mathrm{Tot}}}

\begin{document}

\title{A spectral sequence for parallelized persistence}


\author{David Lipsky}
\address{Department of Mathematics, University of Pennsylvania, 209 South 33rd
  Street. Philadelphia, PA 19104,  USA}
\email{\url{dlipsky@math.upenn.edu}}

\author{Primoz Skraba}
\address{Artificial Intelligence Laboratory, Jo\v zef Stefan Institute, 
Jamova 39, 1000 Ljubljana, Slovenia}
\email{\url{primoz.skraba@ijs.si}}

\author{Mikael Vejdemo-Johansson}
\address{Corresponding author\\School of Computer Science, University of St Andrews, Jack
  Cole Building \\ North Haugh, St Andrews KY16 9SX, Scotland, UK}
\email{\url{mvvj@st-andrews.ac.uk}}


\begin{abstract}
We approach the problem of the computation of persistent homology
for large datasets by a divide-and-conquer strategy. Dividing the
total space into separate but overlapping components, we are able
to limit the total memory residency for any part of the
computation, while not degrading the overall complexity
much. Locally computed persistence information is then merged
from the components and their intersections using a spectral
sequence generalizing the Mayer-Vietoris long exact sequence.

We describe the Mayer-Vietoris spectral sequence and give details
on how to compute with it. This allows us to merge local
homological data into the global persistent
homology. Furthermore, we detail how the classical topology
constructions inherent in the spectral sequence adapt to a
persistence perspective, as well as describe the techniques from
computational commutative algebra necessary for this extension.

The resulting computational scheme suggests a parallelization
scheme, and we discuss the communication steps involved in this
scheme. Furthermore, the computational scheme can also serve as a
guideline for which parts of the boundary matrix manipulation
need to co-exist in primary memory at any given time allowing for
stratified memory access in single-core computation. The spectral
sequence viewpoint also provides easy proofs of a homology
nerve lemma as well as a persistent homology nerve lemma. In
addition, the algebraic tools we develop to approch persistent
homology provide a purely algebraic formulation of kernel, image
and cokernel persistence~\cite{cohen2009persistent}.
\end{abstract}

\maketitle

\begin{figure}[h!]
  \begin{minipage}{1.3\linewidth}
    \subfloat[Double complex]{
      \begin{tikzpicture}[xscale=1.2]
        \foreach \x in {0,...,3} { \foreach \y in {0,...,3} { \node
            (E\x\y) at (\x,\y) {$E^0_{\x,\y}$}; } } \foreach \x/\xx in
        {0/1,1/2,2/3} { \foreach \y in {0,...,3} { \path [->]
            (E\xx\y.west) edge node [auto,swap] {\scriptsize
              $\partial$} (E\x\y.east); } } \foreach \y/\yy in
        {0/1,1/2,2/3} { \foreach \x in {0,...,3} { \path [->]
            (E\x\yy.south) edge node [auto] {\scriptsize $d$}
            (E\x\y.north); } } \draw (E03.north west) -- (E00.south
        west) -- (E30.south east); \foreach \y in {0,...,3} { \node
          (E\y) at (-1.2,\y) {$C_n(\Xcpx)$}; } \foreach \y in
        {0,...,3} { \draw [->] (E0\y.west) -- (E\y.east); }
      \end{tikzpicture}
    }\quad \subfloat[Sphere with a bump]{
      \includegraphics[width=4.5cm]{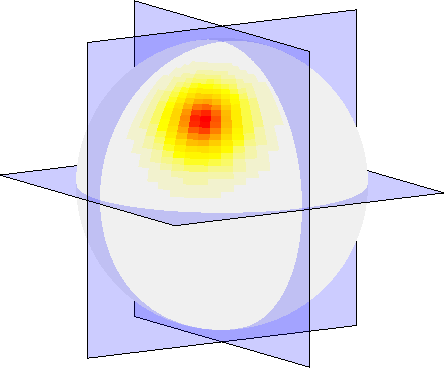}
    }\quad \subfloat[Final computation]{
      \begin{tikzpicture}[xscale=1.1]
        \foreach \x in {0,...,3} { \foreach \y in {0,...,3} { \node
            (E\x\y) at (\x,\y) {$E^\infty_{\x,\y}$}; } } 
       \foreach \y in {0,...,3} { \node (H\y)
          at (-1,\y+1) {$H_\y(\Xcpx)$}; \draw [blue,dotted] (H\y.south
          east) -- (\y,0); }
   \end{tikzpicture}
 }
\end{minipage}
\end{figure}



\clearpage

\section{Introduction}
\label{sec:introduction}


Topological data analysis~\cite{c09TandD,ghrist2008barcodes}
looks to quantify the qualitative aspects of the geometry of a
data set.  The existence and location of voids in an otherwise
densely present figure are among the topological features we can
try detect. The formalism of \emph{homology} allows us to do
this. It has had a large impact: both in geometry and mathematics
in general, and more recently in a computational topology and
geometry setting.  The key to the success of homology in
computational geometry, as well as in other application fields --
is the introduction in~\cite{elz2000} of \emph{persistent
  homology}, a multi-scale approach to the computation of
homological features given a point cloud of samples. 


The original persistence algorithm~\cite{elz2000} has a worst
case, a $O(n^3)$ running time . There has been considerable
effort looking for more efficient algorithms for computing
persistence. For special cases, such as
2-manifolds~\cite{elz2000}, very efficient, nearly linear
algorithms exist. In the general case, recent progress has shown
that persistence can be computed in matrix multiplication
time~\cite{milosavljevic2010zigzag} or alternatively in an output
sensitive way~\cite{chen2011output}.  These running times are
given in terms of simplices, so another active line of research
is into using discrete Morse-theoretic methods or homotopical
collapses~\cite{mmp-gachcm-05,mpz-habas-08,mb-cha-09,mw-chaiph-11,hmmnwjd-ehadmtc-10,dkmw-charc-11}
to reduce the size of the underlying complex while ensuring the
resulting homology computation gives the correct result. These
are particularly useful when the underlying complex is described
by a simpler object such as its
graph~\cite{z2010vr,z2010tidy,als-edsrsschd-11} (i.e. for flag
complexes).

The standard algorithm is most often used, since in practice, it
has almost linear runtime. The algorithm depends on having random
access to a representation of the boundary operator. The standard
representation is as a sparse matrix and during the course of the
standard algorithm in higher dimensions, the matrix quickly
fills-in requiring $O(n^2)$ storage space. In practice, the
requirement of random access to such a large data structure has
proven to be the limiting factor in computation. There has been
work on distributed computation, primarily focusing on ordinary
homology in sensor networks and using either coreduction
techniques as above~\cite{djmg-dccsnhm-11} or via combinatorial
Laplacians~\cite{mj-dchgg-07}.  The latter approach does not
readily extend to persistent homology and is much slower in
practice than centralized methods.



In this paper, we present an algorithm which uses a
divide-and-conquer approach to computing both homology and
persistent homology. Rather than computing on the entire complex,
we break the computation up into smaller pieces. The benefit
being that we can compute persistence independently on the each
of the pieces or patches. The key contirbution of
this paper is to show how to merge the results from the
individual patches into a global quantity\footnote{The merging
  step is what has made this kind of approach difficult in the
  past, particularly for persistence}. Our approach is based on a
tool from algebraic topology called \emph{spectral
  sequences}. Dividing the space up based on a \emph{cover}
$\Ucoll$, the spectral sequence gives the machinery necessary to
iteratively compute better approximations of the persistent
homology of a total space $\Xcpx$ from the local persistence
computations in each $\Ucpx_i\in\Ucoll$.

We focus our attention on the underlying ideas based
on spectral sequences and the resulting algorithms rather than on
the analysis of these algorithms, which depend heavily on the
cover (i.e. how we break up the underlying complex). We partially
address the problem of computing covers by presenting some
straightforward approaches which could be used. Using the
developed techniques also gives a more convenient form of the
Nerve~\cite{hatcher2005algebraic} and Persistent
Nerve~\cite{co-tpbr-08} lemmas which are widely used in
persistent homology, as well as an algebraic description of the
image, kernel and cokernel persistence introduced
in~\cite{cohen2009persistent}.

The paper is structured as follows: we first introduce the
Mayer-Vietoris spectral sequence, which is the main algebraic
tool that will be used in the paper. We adapt the treatement
of~\cite{godement1958topologie} to a persistent homology context,
and provide mathematical descriptions of the spectral sequence
computation. Using these structural results, in
Section~\ref{sec:algorithm} we give explicit descriptions of all
the needed algorithms; from basic algebraic primitives computing
kernels, images and cokernels of maps between finitely presented
$\kk[t]$-modules to the concrete computation of the spectral
sequence.  Finally, in Section~\ref{sec:applications} we discuss
several applications and corollaries to the theory and
algorithmics in this paper. In particular, we discuss strategies
for parallelizing persistence and for computing persistent
homology with a stratified memory residence approach; 
the relationship between our algebraic primitives and the
corresponding geometric primitives in~\cite{cohen2009persistent};
and we give proofs of the homological Nerve lemma and of the
persistent homological Nerve lemma that avoid the requirement of
contractibility of all the covering patches and patch
intersections.


\subsection{Background}
\label{sec:background}
For an introduction to persistent homology we refer the reader
to~\cite{eh-ct-09,z2005} along with~\cite{hatcher2005algebraic}
for an introduction to homology within algebraic topology.  The
main object which we look to compute is the \emph{persistence
  module}. This is an algebraic object introduced
in~\cite{cz2005}, which encodes the homology of a
\emph{filtration}, a sequence of increasing simplicial complexes
\begin{equation*}
\emptyset = \Xcpx_0 \subseteq\Xcpx_1 \subseteq\ldots \subseteq\Xcpx_{n-1} \subseteq \Xcpx_N  = \Xcpx
\end{equation*}
The module encodes the information in all pairwise maps between
these spaces. It turns out that, this information can be
represented by a \emph{barcode} or \emph{persistence diagram}
which are a set of coordinates which represent the \emph{birth}
and \emph{death} times of non-trivial homology classes in the filtration. 

These have geometrical meaning if the filtration is derived from
some geometric quantity. Common examples include: the sub-level
set filtration with respect to some function defined on the
space and the filtration based on distances between points in an
input point cloud.

In this paper, we compute homology over coefficients in a field
$\kk$. In this case, the homology over a space is a vector
space. However, as shown in~\cite{cz2005}, rather than computing
the homology for each space independently, we can build graded
$\kk[t]$-linear chain complexes where the degree of each
generator encodes the entry of a simplex into the
filtration. Refering to this as \emph{persistence complex}, the
authors show that persistent homology is equivalent to ordinary
homology on the persistence complex with coefficients in
$\kk[t]$.

The division of our space will be in terms of a cover. A cover of
a space $\Xcpx$ is an indexed family of sets $\Ucpx_i$ such that
$\Xcpx \subseteq \cup_i \Ucpx_i $.  The methods we develop work
with arbitrary covers of the space, making the method suitable
for any division of the underlying space (of course, the
performance will depend on the properties of the cover). Much of
the exposition is given in the language of commutative
algebra.  This not only simplifies the exposition, but
gives important insight into how the algorithms should
work. Whereever possible, we try to put things into the context
of previous approaches.




\section{The Mayer-Vietoris spectral sequence}
\label{sec:mayer-viet-spectr}

Spectral sequences have seen numerous uses in classical algebraic
topology \cite{mccleary2001user}. In the field of persistent
homology, an awareness of the role of spectral sequences has been
present from the start \cite{cz2005}, but the difficulty in
computing with spectral sequences has consistently made
alternative routes attractive. We do not give a complete
background on spectral sequences here but rather, we focus our
attention on one particular spectral sequence, which we use for
the parallelization of persistent homology. For reference,
spectral sequences are discussed extensively
in~\cite{mccleary2001user} although more accessible introductions
can be found in~\cite{chow2006you}
and~\cite{hatcher2005spectral}.

The spectral sequence we use can be viewed as generalization of the
Mayer-Vietoris long exact sequence which given two pieces, relates the
homology of their union with the homology of the pieces and their
intersection.  The same way that the long exact sequence encodes an
inclusion-exclusion principle on homology classes, the corresponding
spectral sequence provides the mechanics of an inclusion-exclusion
principle for deeper intersections than the pairwise afforded by the
long exact sequence.  Frequently mentioned in the literature (inter
alia: \cite{mccleary2001user}), we found the most detailed expositions
in
\cite{bott1982differential,brown1982cohomology,godement1958topologie}.
The spectral sequence builds up a space very similar to the blowup
complex introduced by Carlsson and Zomorodian in
\cite{cz2007localized}, but uses more of the available
information in order to knit
together local information to a global homology computation. Also Carlsson
and Zomorodian do not use the same filtration as we depend on here.
The following is based very loosely on the exposition in
\cite{godement1958topologie}.

As in \cite{cz2007localized}, we consider a filtered simplicial
complex $\Xcpx$, with a cover $\Ucoll$ by filtered subcomplexes $\Xcpx =
\bigcup_{i\in\Icpx}\Ucpx_i$ (see Figure \ref{fig:blowup}(b)).
The cover can of course be anything, but in practice we wish to divide
the full filtered complex into almost disjoint pieces so as to
maximally parallelize the computation. The condition of covering with
filtered subcomplexes is trivially fulfilled if we construct the cover
on the final space of the filtration.

We recall the definition of the blowup complex from \cite{cz2007localized} 
\[
\Xcpx^{\Ucoll} = \bigcup_{\emptyset\neq J\subseteq\Icpx}
\left(\bigcap_{j\in J} \Ucpx_j\right)\times\Delta^J
\]
where $\Delta^J$ is the $|J|$-simplex with vertices numbered from
$J$.

Fundamentally, the question is how we can compute the homology of
the blowup complex. We approach this by separating out the
problem into computational components. Whenever $n$, $p$ and $q$
occur simultaneously, we shall assume that $n=p+q$. The $q$-chains
in a $p$-fold intersection of the original cover generate
$p+q$-chains in the chain complex of the blowup complex, and thus
potentially $p+q$-homology.

\begin{figure}
  \centering
  \subfloat[An underlying topological space covered by three sets.]{
    \includegraphics[width=3.7cm,page=10]{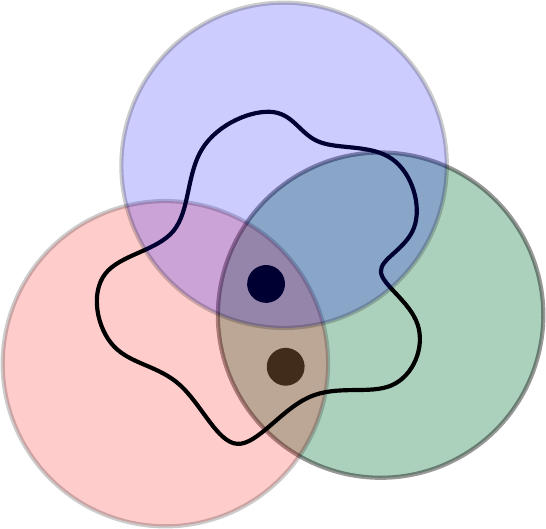}}\qquad
  \subfloat[The sets shown as disjoint spaces.]{
    \includegraphics[width=3.7cm,page=7]{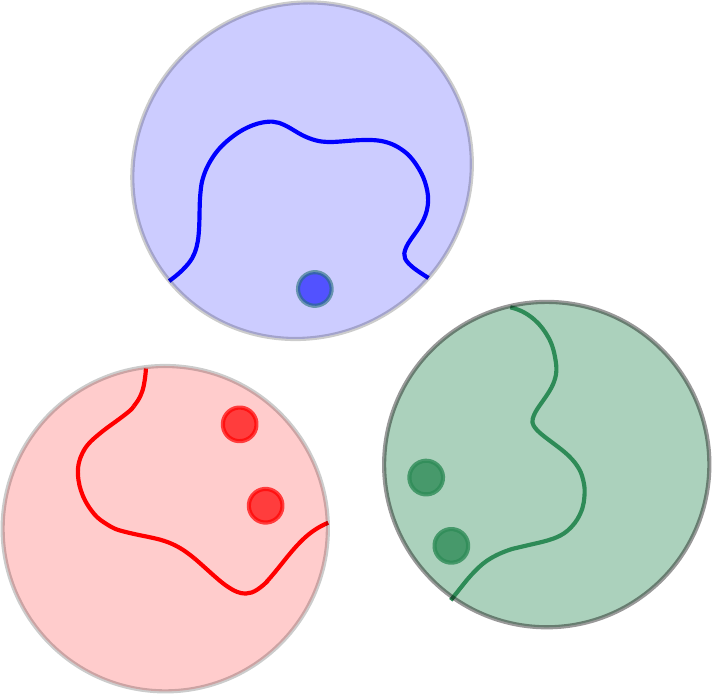}}\qquad
  \subfloat[Blowup complex.]{
    \includegraphics[width=4.2cm,page=12]{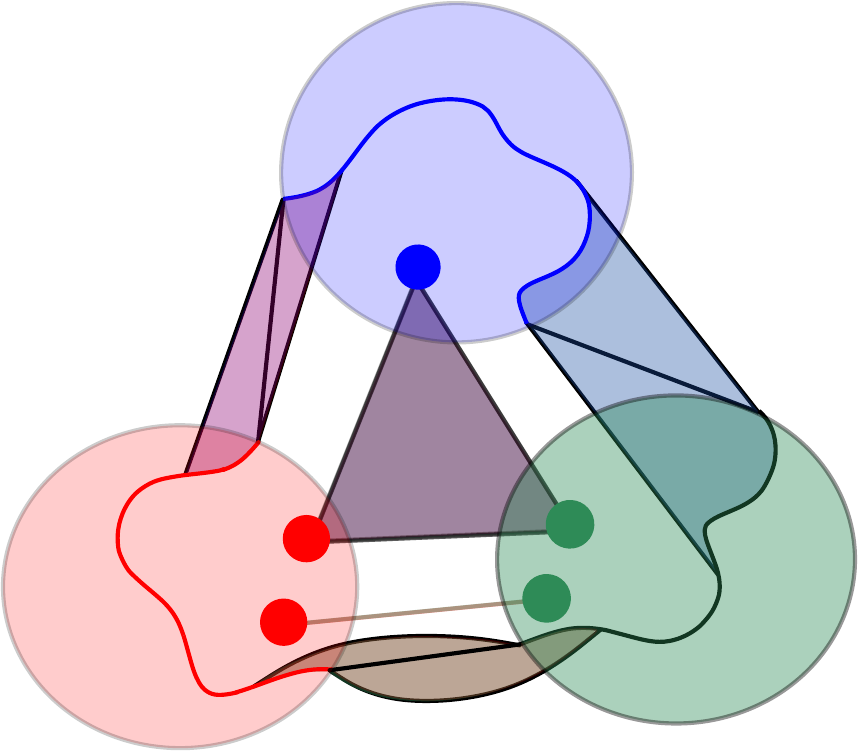}}\qquad
  \caption{\small Construction of the blowup for a triple cover. In the
    blowup (c), anything in a double intersection is
  multiplied with an edge, and anything in a triple intersection is
  multiplied by a 2-simplex. While only visible in the top right one,
  all three parts of the curve going through a double intersection
  will be replaced by a triangulated quadrilateral.}
  \label{fig:blowup}
\end{figure}

We use this to organize our computation. Let us write $E^0_{p,q}$
for the $q$-chains in a $p$-fold intersection of the original
cover\footnote{The $E^0_{p,q}$ notation will be familiar to
  anyone who has worked with spectral sequences as the 0-page of
  the spectral sequence}. The individual chain complexes $E^0_{p,*}$ have
boundary maps inherent from their simplicial complex structure,
and we shall write 
$d^0_{p,q}$ for the map $E^0_{p,q}\to E^0_{p,q-1}$, and 
$d^0$ for the aggregate map $E^0_{*,*}\to E^0_{*,*}$.

\begin{figure}
  \subfloat[Double complex expression of the $E^0$ page.]{
    \begin{tikzpicture}[xscale=1.3,yscale=1.1]
      \foreach \x in {0,...,4} { 
        \foreach \y in {0,...,3} { 
          \node (E\x\y) at (\x,\y) {$E^0_{\x,\y}$}; 
        } 
      } 
      \foreach \x/\xx in {0/1,1/2,2/3,3/4} { 
        \foreach \y in {0,...,3} { 
          \path [->] 
          (E\xx\y.west) 
          edge node [auto,swap] {\scriptsize $d^1$} 
          (E\x\y.east); 
        } 
      } 
      \foreach \y/\yy in {0/1,1/2,2/3} {
        \foreach \x in {0,...,4} { 
          \path [->] (E\x\yy.south) 
          edge node [auto] {\scriptsize $d^0$} 
          (E\x\y.north); 
        } 
      }
      \draw (E03.north west) -- (E00.south west) -- (E40.south east);
      \foreach \y in {0,...,3} {
        \node (E\y) at (-1,\y) {$C_n(\Xcpx)$};
     }
     \foreach \y in {0,...,3} {
        \draw [->] (E0\y.west) -- (E\y.east);
     }
    \end{tikzpicture}
  }\hfill
  \subfloat[Typical element of the total complex.]{
    \begin{tikzpicture}[xscale=1.3,yscale=0.9]
      \foreach \x in {0,...,4} {
        \foreach \y in {0,...,4} {
          \node (E\x\y) at (\x,\y) {\phantom{$E^0_{0,0}$}}; 
        }
      }
      \node (E4) at (-1,4) {$\alpha$};
      \draw (E04.north west) -- (E00.south west) -- (E40.south east);
      \foreach \x/\y/\xx in {0/4/,1/3/0,2/2/1,3/1/2,4/0/3} {
        \node at (\x,\y) {$\alpha_\x$};
        \draw [->] (E\x\y.west) -- (E\xx\y.east);
      }
      \foreach \x/\y/\yy in {0/3/4,1/2/3,2/1/2,3/0/1} {
        \node at (\x,\y) {$\omega_\x$};
        \draw [->] (E\x\yy.south) -- (E\x\y.north);
      }
   \end{tikzpicture}
  }
  \caption{\small Double complex setup. A double complex is a grid as
    depicted in (a), with modules at each grid point, and two families
    of maps $d^0$ and $d^1$ between the modules. For a double
    complex, we require $d^0d^1+d^1d^0=0$ as well as $d^0d^0=0$
    and $d^1d^1=0$. A typical element of the total complex is some
    tuple $(\alpha_0,\dots,\alpha_n)$, and will map under the total
    differential to some other element $(\omega_0,\dots,\omega_{n-1})$
    given by summing up the contributions from each direction.}
\label{fig:doublecomplex}
\end{figure}
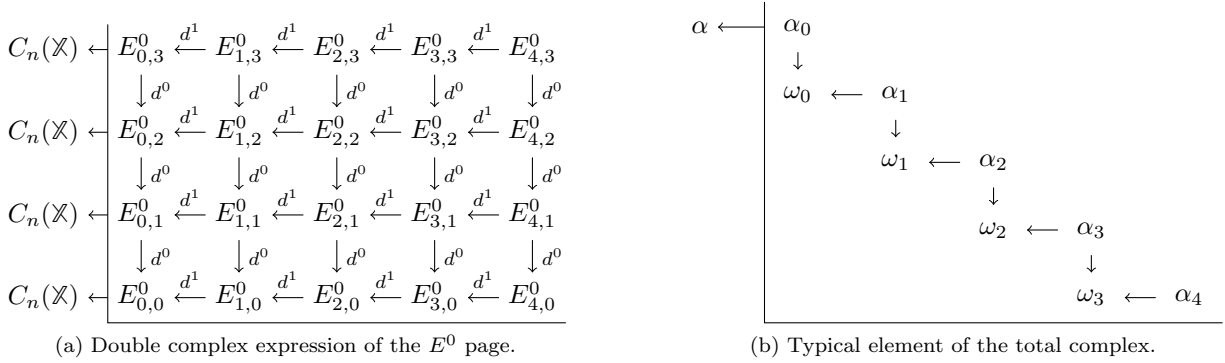

For any particular $J\subseteq\Icpx$, and any particular $j'\in J$,
there is an inclusion map $\bigcap_{j\in J}\Ucpx_j\to\bigcap_{j\in
  J\setminus j'}\Ucpx_j$. If we multiply each such inclusion map with
the sign of the corresponding term $J\to J\setminus j'$ of the nerve
complex boundary map, and sum the maps up for all $j'$, we get a map
$d^1_{p,q}: E^0_{p,q}\to E^0_{p-1,q}$. Aggregating over all chains, we
have a map $d^1: E^0_{*,*}\to E^0_{*,*}$.

A structure like this, with bi-indexed modules $E^0_{p,q}$ and
horizontal maps $d^0$ and vertical maps $d^1$ such that $d^0\circ
d^0=0$, $d^1\circ d^1=0$ and $d^0\circ d^1 + d^1\circ d^0=0$ is called
a \emph{double complex}. For double complexes, we can generate a
straightforward chain complex called the \emph{total complex} by
writing $\Tot(E^0_{*,*})_n = \bigoplus_{p+q=n} E^0_{p,q}$, and
introduce as a boundary map the aggregate map $d^0+d^1$.

For our computation of the homology of the blowup complex, we rely on
a fundamental theorem from algebraic topology
\cite{godement1958topologie}:

\begin{theorem}\label{thm:isomorphism}
  $H_*(\Tot(E^0_{*,*})) = H_*(\Xcpx)$.
\end{theorem}

While we shall not prove the theorem in its full extent here, we still
need the statement of the isomorphism for the development of our
algorithm. A chain in $E^0_{1,q}$ is a collection of $q$-chains in the
covering patches. Such a chain can be mapped into $C_*\Xcpx$ by simply
summing up the component chains, included into the full complex. Just
summing up these components may both introduce and remove homological
features present in the composite chain -- and the exact behaviour of
these introductions and removals is encoded into the other summands of
the direct sum. Thus, the isomorphism identifies a class 
$[(\alpha_0,\dots,\alpha_n)]$ with a class $[\alpha]$ precisely
if $\alpha$ is the result of summing up the components of
$\alpha_0$. The chains $\alpha_1,\dots,\alpha_n$ capture, in a way
reminiscent of the inclusion-exclusion principle, the exact ways in
which $\alpha$ differs from merely summing up $\alpha_0$.

The columns in the double complex provide a filtration on
$\Tot(E^0_{p,q})$. We write $E^0_{\leq p,*}$ for the subcomplex
consisting of the $p$ leftmost columns, and write
$L_{n,p}=\img(H_n(\Tot(E^0_{\leq p,*}))\to H_n(\Tot(E^0_{*,*})))$ for
the image of the map induced by the inclusion of a particular
filtration step into the entire double complex. Since the $E^0_{\leq
  p,*}$ filter $E^0_{*,*}$, the $L_{n,p}$ filter
$H_n(\Tot(E^0_{*,*}))$. Furthermore, a single class
$[(\alpha_0,\dots,\alpha_n)]$ is in $L_{n,p}$ if and only if the class
has a representative on the form
$(\alpha'_0,\dots,\alpha'_p,0,\dots,0)$.

While computing $L_{n,p}$ completely ends up being similar in
difficulty to the original homology computation, we can consider
the computation of just the elements that get added at any
particular step of the filtration. In other words, if we compute
$L_{n,p}/L_{n,p-1}$, this tells us about only those elements of
$H_*(\Xcpx)$ that show up at precisely the $p$-fold
intersections. Clearly, taking the direct sum over all $p$ of
these quotients, we can compute a complete description of the
entire homology $H_*(\Xcpx)$.

Given some element $[(\alpha_0,\dots,\alpha_p,0,\dots,0)]\in L_{n,p}$,
the equivalence class it belongs to in $L_{n,p}/L_{n,p-1}$ is
completely determined by the value of $\alpha_p$. The particular
$\alpha_p$ representing such an element is not uniquely determined,
nor is every element of $E^0_{p,q}$ a valid choice for some element
determining an equivalence class of $L_{n,p}/L_{n,p-1}$. We shall
write $Z^\infty_{p,q}$ for the $\alpha_p$ that do represent some
equivalence class of $L_{n,p}/L_{n,p-1}$, and we shall write
$B^\infty_{p,q}$ for the elements of $E^0_{p,q}$ that represent the 0
class of $L_{n,p}/L_{n,p-1}$.

The spectral sequence here aims to provide increasingly accurate
approximations to these $Z^\infty_{p,q}$ and $B^\infty_{p,q}$. These
approximations are computed by introducing the conditions defining
$Z^\infty_{p,q}$ and $B^\infty_{p,q}$ one at a time.

\paragraph{\bf The case $Z^r_{p,q}$} 
Fix some $\alpha_p\in E^0_{p,q}$. For $\alpha_p$ to be in
$Z^\infty_{p,q}$, there has to be a sequence
$\alpha_{p-1},\dots,\alpha_0$ as in Figure
\ref{fig:cyclesandboundaries} (a) such that
\begin{equation}\label{eq:zinfinity}
d^0\alpha_p = 0 \qquad 
d^0\alpha_{p-1}=d^1\alpha_p \qquad 
\dots \qquad
d^0\alpha_0 = d^1\alpha_1
\end{equation}

These equations express the condition
$d(\alpha_0,\dots,\alpha_p,0,\dots,0)=0$ in $\Tot(E^0_{*.*})$.

We define $Z^r_{p,q}$ to be the set of all $\alpha_p$ that fulfill the
conditions
\begin{equation}\label{eq:cycle_equations}
d^0\alpha_p = 0 \qquad 
d^0\alpha_{p-1}=d^1\alpha_p \qquad 
\dots \qquad
d^0\alpha_{p-r} = d^1\alpha_{p-r+1}
\end{equation}
and in particular, we can compute $Z^r_{p,q}$ recursively as those
elements $\alpha_p$ of $Z^{r-1}_{p,q}$ for which there in addition to
the $\alpha_{p-1},\dots,\alpha_{p-r+1}$ elements guaranteed to exist
since $\alpha_p\in Z^{r-1}_{p,q}$, there is also an $\alpha_{p-r}\in E^0_{p-r,q+r}$
such that $d^0\alpha_{p-r} = d^1\alpha_{p-r+1}$.
It is clear from these definitions that $Z^p_{p,q}=Z^\infty_{p,q}$.

\paragraph{\bf The case $B^r_{p,q}$}
\label{sec:case-br_p-q}
As above, we fix $\alpha_p\in E^0_{p,q}$. In order for $\alpha_p$ to
be in $B^\infty_{p,q}\subseteq Z^\infty_{p,q}$, we additionally need a
collection of elements $\beta_p,\dots,\beta_{n+1}$ as in Figure
\ref{fig:cyclesandboundaries} (b) such that 
\begin{equation}\label{eq:boundary_equations}
d^0\beta_p+d^1\beta_{p+1} = \alpha_p \qquad
d^0\beta_{p+1}+d^1\beta_{p+2} = 0 \qquad
\dots \qquad
d^0\beta_n+d^1\beta_{n+1} = 0
\end{equation}

The simplest case when this happens is if there is some $\beta_p$ such
that $d^0\beta_p = \alpha_p$. In this case we can choose $\beta_j=0$
for all other $\beta_j$. We write $B^0_{p,q}$ for the vector
space of all such $\alpha_p$.

In general, we shall define $B^r_{p,q}$ to be the set of $\alpha_p$
for which there are $\beta_p,\dots,\beta_{p+r}$ satisfying the equations
\[
d^0\beta_p+d^1\beta_{p+1} = \alpha_p \qquad
d^0\beta_{p+1}+d^1\beta_{p+2} = 0 \qquad
\dots \qquad
d^0\beta_{p+r-1}+d^1\beta_{p+r} = 0 \qquad
d^0\beta_{p+r}=0
\]

In particular $\alpha_p\in B^r_{p,q}$ if $\alpha_p\in B^{r-1}_{p,q}$
or if there are $\beta_p,\dots,\beta_{p+r}$ satsifying all these equations.

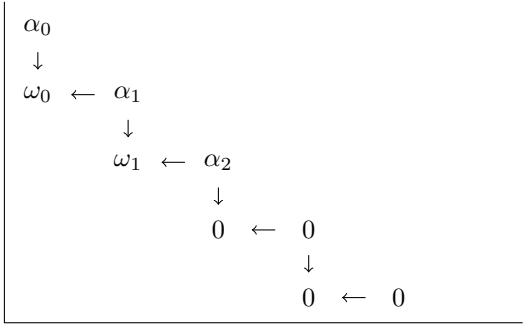
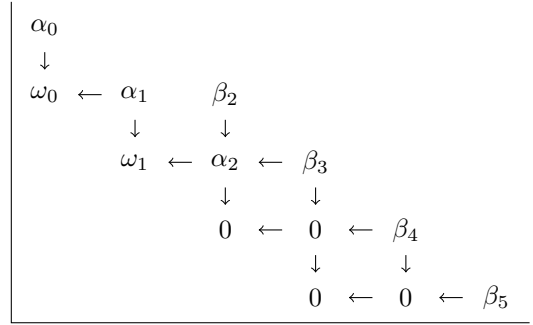
\begin{figure}
  \subfloat[Condition for a cycle:
  $(\alpha_0,\dots,\alpha_2,0,\dots,0)$ is a cycle if there are
  $\omega_0,\omega_1$ such that $d^0\alpha_0=\omega_0=d^1\alpha_1$,
  $d^0\alpha_1=\omega_1=d^1\alpha_2$ and $d^0\alpha_2=0$.]{
    \begin{tikzpicture}[yscale=0.9,xscale=1.2]
      \foreach \x in {0,...,5} {
        \foreach \y in {0,...,4} {
          \node (E\x\y) at (\x,\y) {\phantom{$E^0_{0,0}$}}; 
        }
      }
      \draw (E04.north west) -- (E00.south west) -- (E50.south east);
      \node at (0,4) {$\alpha_0$};
      \foreach \x/\y/\xx in {1/3/0,2/2/1} {
        \node at (\x,\y) {$\alpha_\x$};
        \draw [->] (E\x\y.west) -- (E\xx\y.east);
      }
      \foreach \x/\y/\xx in {3/1/2,4/0/3} {
        \node at (\x,\y) {$0$};
        \draw [->] (E\x\y.west) -- (E\xx\y.east);
      }
      \foreach \x/\y/\yy in {0/3/4,1/2/3} {
        \node at (\x,\y) {$\omega_\x$};
        \draw [->] (E\x\yy.south) -- (E\x\y.north);
      }
      \foreach \x/\y/\yy in {2/1/2,3/0/1} {
        \node at (\x,\y) {$0$};
        \draw [->] (E\x\yy.south) -- (E\x\y.north);
      }
   \end{tikzpicture}
  }\hfill
  \subfloat[Condition for a boundary:
  $(\alpha_0,\dots,\alpha_2,0,\dots,0)$ is a boundary if in addition
  to the $\omega_0,\omega_1$ in (a), there are also
  $\beta_2,\dots,\beta_5$ such that $d^0\beta_2+d^1\beta_3=\alpha_2$,
  $d^0\beta_3+d^1\beta_4=0$, $d^0\beta_4+d^1\beta_5=0$ and (trivially)
  $d^0\beta_5=0$.]{
    \begin{tikzpicture}[yscale=0.9,xscale=1.2]
      \foreach \x in {0,...,5} {
        \foreach \y in {0,...,4} {
          \node (E\x\y) at (\x,\y) {\phantom{$E^0_{0,0}$}}; 
        }
      }
      \draw (E04.north west) -- (E00.south west) -- (E50.south east);
     \foreach \x/\y/\yy in {0/4/5,1/3/4,2/2/3} {
        \node at (E\x\y) {$\alpha_\x$};
     }
      \draw [->] (E23.south) -- (E22.north);
      \foreach \x/\y/\yy in {3/1/2,4/0/1} {
        \node at (E\x\y) {$0$};
        \draw [->] (E\x\yy.south) -- (E\x\y.north);
      }
     \node at (E23) {$\beta_2$};
      \foreach \x/\xx/\y in {3/2/2,4/3/1,5/4/0} {
        \node at (E\x\y) {$\beta_\x$};
        \draw [->] (E\x\y.west) -- (E\xx\y.east);
      }
      \foreach \x/\y/\yy in {0/3/4,1/2/3} {
        \node at (\x,\y) {$\omega_\x$};
        \draw [->] (E\x\yy.south) -- (E\x\y.north);
      }
      \foreach \x/\y/\yy in {2/1/2,3/0/1} {
        \node at (\x,\y) {$0$};
        \draw [->] (E\x\yy.south) -- (E\x\y.north);
      }
      \foreach \x/\y/\xx in {1/3/0,2/2/1,3/1/2,4/0/3} {
       \draw [->] (E\x\y.west) -- (E\xx\y.east);
      }
   \end{tikzpicture}
  }
\caption{\small Cycles and boundaries in $E^r_{p,q}$.}
\label{fig:cyclesandboundaries}
\end{figure}

\paragraph{\bf $Z^r_{p,q}$ and $B^r_{p,q}$ as kernel and image of a differential}
\label{sec:zr_p-q-br_p}
Suppose $\alpha_p\in Z^r_{p,q}$. Then there exist
$\alpha_{p-1},\dots,\alpha_{p-r}$ as above. Let
$\omega_{p-r-1}=d^1\alpha_{p-r}\in Z^r_{p-r-1,q+r}$. This element is
well-defined, depending only on $\alpha_p$, up to a sequence of error
terms for all the $\alpha_{p-i}$, each of which is guaranteed to be in
the appropriate $B^r_{p-i,q+i}$. Hence, the map
$[\alpha_p]\mapsto[\omega_{p-r-1}]$ is a well-defined map
$d^r: Z^r_{p,q}/B^r_{p,q}\to Z^r_{p-r-1,q+r}/B^r_{p-r-1,q+r}$. This
collection of quotient modules is often denoted $E^r_{p,q} =
Z^r_{p,q}/B^r_{p,q}$ and called the $E^r$-page. 

The elements $\alpha_p,\dots,\alpha_{p-r}$ are an appropriate set of
$\beta_j$ to witness $\omega_{p-r-1}\in B^r_{p-r-1,q+r}$. Re-indexing,
we can conclude that $\alpha_p\in B^{r+1}_{p,q}$ if and only if
$[\alpha_p]=d^r[\beta_{p+r+1}]$ for some $\beta_{p+r+1}\in
Z^r_{p+r+1,q-r}$.

Finally, if $w_{p-r-1}\in B^r_{p-r-1,q-r}$, or equivalently if
$d^r[\alpha_p] = [0]$, then it follows that $\alpha_p\in Z^r_{p,q}$.

All these statements can be easily proven by diagram chases. 

If at any stage of the computation, all the maps $d^r:E^r_{*,*}\to
E^r_{*,*}$ are the zero-map, then the corresponding $Z^r_{*,*}$ and
$B^r_{*,*}$ already satisfy all their conditions, and so for that $r$,
it already holds that $Z^r_{*,*}=Z^\infty_{*,*}$ and
$B^r_{*,*}=B^\infty_{*,*}$. Many proofs in algebraic topology work by
proving that for a low value of $r$, the map $d^r$ vanishes,
eliminating the need to construct the higher index maps. In Section
\ref{sec:nerve-lemma} we shall see an application of this principle.

\section{Algorithms}
\label{sec:algorithm}

We shall state algorithms here both for the underlying algebraic
operations, as well as for the resulting spectral sequence and
persistence computations.

\subsection{Algebraic operations}
\label{sec:algebraic-operations}

In algebra textbooks written with a sufficiently general approach
to linear algebra, the development of matrix arithmetic and its
uses in solving linear equations will be done in a manner
appropriate for a generic \emph{Euclidean
  domain}~\cite{sims1984abstract}\footnote{This is a slightly
  more restrictive class than principle ideal domains.}. It is a
well-known fact that in addition to any field, further Euclidean
domains include $\mathbb Z$, as well as $\kk[t]$ for any field
$\kk$. In this setting, a few useful facts hold: any
submodule of a free module is free, and Gaussian elimination is an
appropriate algorithm to compute matrix ranks, function kernels and
images, reducing bases to remove redundancies and to form appropriate
reduction systems.

In light of this, we draw the reader's attention to the
implications for persistence. A persistence chain complex is a free graded
$\kk[t]$-module with a generator in degree $k$ for each simplex that
enters the filtered chain complex in the $k$th filtration
step. Persistent homology is a finitely presented graded
$\kk[t]$-module, with generators given by a cycle basis and relations
given by a boundary sub-basis of the cycle basis. 

We can assume -- especially in the topological case -- that if
some $t^k\sum\lambda_jt^jm_j$ is a boundary basis element, then
the cycle was a cycle already $k$ time-steps earlier, represented
by $\sum\lambda_jt^jm_j$.  We shall represent finitely presented
modules precisely as a segregated pair of bases: a generators
basis and a separate relations basis. Because of this nice
behaviour of the generators corresponding to specific
relations, we can further allow ourselves to keep the relations
basis in a row-reduced form -- to ensure that the division
algorithm produces nice normal forms of all elements -- and the
generators basis reduced with respect to the relations.

As the reader recalls, Gaussian elimination puts a matrix into a
normal form by clearing out row by row by using the leftmost
column in the matrix with no non-zero entries above of the
current row to cancel out any entries occurring later. When
working with graded $\kk[t]$-modules, some care has to be taken
to avoid reducing any columns using columns that do not exist in an
early enough grading. This can be accomplished by sorting the
columns in the matrix in order of rising degree, and only reducing
to the right. Tracking the operations performed during a reduction
allows for the computation of a kernel by reading off the
operations that produced all-zero columns in the reduced matrix.

Given a function $f: M/Q\to N/P$, where $M/Q$ and $N/P$ are both
represented by lists of generators and lists of relations as discussed
above, the following three constructions come naturally:

\paragraph{\bf Image modules}
\label{sec:image-modules}
To compute $\img f$, we compute the list $f(m)$ for $m\in M$. Since
$f(Q)\subseteq P$, the relations on $\img f$ are the relations listed
in $P$, and so $\img f=f(M)/P$.

\paragraph{\bf Cokernel modules}
\label{sec:cokernel-modules}
To compute $\coker f$, we need to compute $(N/P)/(f(M/Q))$. Again,
$f(Q)\subseteq P$, and thus we get the resulting module by imposing
all the images $f(m)$ of the generators of $M$ as additional relations.
Hence, $\coker f=N/(f(M)\cup Q)$.

\paragraph{\bf Kernel modules}
\label{sec:kernel-modules}
The kernel is, just as in~\cite{cohen2009persistent}, more
complex than the image and cokernel. The computation of the
kernel module as well as its embedding into $M/Q$ is computed in
a two-step process. An element of $\ker f$ is going to be
represented by some element of the free module $M$ such that
$f(m)\in P$. We can compute this by computing the kernel $F_K$ of
the map $M\oplus P\to N$ given by $(m,p)\mapsto f(m)-p$. This is
a kernel of a map between free $\kk[t]$-modules, and thus
computable with the matrix algorithm above. The actual generators
in $M$ are given by the projection onto the first summand $\pi_M:
(m,q)\mapsto m$, and we write $K=\pi_M(F_K)$.

The resulting elements form a free submodule of $M$ giving generators
of the kernel module. The kernel module will have relations, however,
and not all relations in $Q$ are going to be in the submodule
$K$. Hence, for a full presentation, we need to find the relations
that are actually in $K$, which we can do with another kernel
computation. We consider the map $K\oplus Q\to M$ given by
$(k,q)\mapsto k-q$ and compute its kernel. The projection $\pi_K:
(k,q)\mapsto k$ gives us a set of relations.

The constructions above, as well as Gr\"obner basis approaches to
handling free modules over polynomial rings are described in
\cite{eisenbud1995commutative}.

\subsection{Spectral sequence differentials}
\label{sec:spectr-sequ-diff}
Using the algebraic operations described above, we describe an
algorithm for computing persistent homology via the Mayer-Vietoris
spectral sequence. This algorithm is iterative and \emph{converges}
to the correct cycles $Z^\infty$ and boundaries $B^\infty$ so that the
isomorphism in Theorem \ref{thm:isomorphism}  holds.

Converge occurs when the spectral sequence collapses.  In
particular, if the differential $d^r$ is the zero-map everywhere,
then $Z^{r+1}=Z^r$, and $B^{r+1}=B^r$, and we can verify that
$d^{r+1}$ is also going to be the zero-map everywhere, and this
continues all the way up. One powerful use of this is that once
the maps $d^r$ traverse the double complex in such a way that
they only hit trivial modules, the corresponding $d^r$ has to be
the zero-map.  This happens, for instance, if the double complex
is contained in a $p\times q$ bounding box where $p$ and $q$ are
maximal dimensions of the complex and nerve respectively.  In
this case, the spectral sequence will collapse for
$r>\min(p,q)$. At that point all the maps will map outside the
box and so can only hit trivial modules.

We first describe the initialization of the algorithm. Beginning with
the double complex of Figure~\ref{fig:doublecomplex}(a), we
first compute homology vertically, replacing each $E^0_{p,q}$ with the
quotient module $Z^1_{p,q}/B^1_{p,q}$ where $Z^1_{p,q}=\ker
d^0_{p,q}$ and $B^1_{p,q}=\img d^0_{p,q+1}$. This computes, in effect,
the persistent homology of the chain complex over each simplex in the
nerve of the cover. We represent $E^1_{p,q}$ by its segregated basis,
as described in Section~\ref{sec:algebraic-operations}, maintaining
the free modules $Z^1_{p,q}$ and $B^1_{p,q}$ as well as the preboundary
$I^1_{p,q}$ defined by the equation $d^0_{p,q}I^1_{p,q}=B^0_{p,q-1}$.

Next, we describe the iterative step. We will, for now, assume that we
are given the differentials and describe their computation afterwards.

At the $r$-th iteration, at each node $(p,q)$, we compute the kernel
module of $d^r_{p,q}$, storing it as $Z^r_{p,q}$ and the image module
of $d^r_{p+r-1,q-r}$, storing this as $B^r_{p,q}$. Since $d^r$ is a
differential, we know that $B^{r+1}_{p,q}\subseteq Z^{r+1}_{p,q}$, and by
reducing both generating sets we can find a good presentation of the
quotient $Z^{r+1}_{p,q}/B^{r+1}_{p,q}=E^{r+1}_{p,q}$. 

As for the concrete maps $d^r$, we compute them as we go along. In
particular, $d^1$ is the induced map on $E^1$ by giving the inclusion
maps of deeper intersections into more shallow the coefficients from
the nerve boundary map. The higher differentials $d^r$ are computed
according to the machinery of Section~\ref{sec:mayer-viet-spectr}. In
particular, the images of $d^{r-1}_{p,q}$ are retained. For an element
$x\in E^r_{p,q}$ in the kernel of $d^{r-1}_{p,q}$ we can find a
corresponding image element $d^{r-1}_{p,q}x \in
B^0_{p-r+2,q+r-1}$. Therefore, we can find some $y\in I^1_{p-r+2,q+r}$
such that $d^0_{p-r+2,q+r}y=d^{r-1}_{p,q}x$. We define $d^rx =
d^1y$. This corresponds precisely to the description in
Section~\ref{sec:mayer-viet-spectr}, where the theoretical motivation
can be found explaining exactly why these operations are all possible
to perform and appropriate. In each iteration, we use the image and
kernel module computations between finitely presented modules and maps
that are computed in earlier stages.

For a lower level description of the algorithm see Appendix~\ref{sec:implementation}.

\section{Explicit example}
\label{sec:explicit-example}

Consider the sphere in Figure~\ref{fig:explicit_example}(a). Its function will introduce the
hot spot and its surrounding area significantly later than the
remaining sphere. Computing persistence over the entire sphere, we would expect to see a signature somewhat
like the one in Figure~\ref{fig:explicit_example}(c).
\begin{figure}
\subfloat[]{
      \includegraphics[width=4.5cm]{spherepeak}
}\hfill
\subfloat[]{
      \includegraphics[width=3cm]{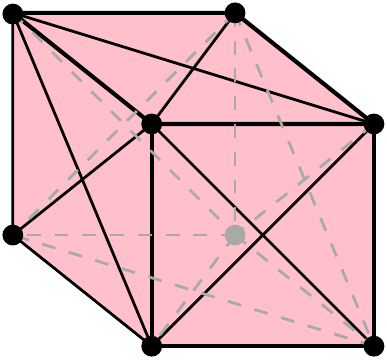}
}\hfill
\subfloat[]{
\begin{tikzpicture}
  \begin{scope}[yshift=1cm]
    \draw [->] (0,0) -- (0,1); 
    \draw [->] (0,0) -- (4,0); 
    \node at (-0.5,0.5) {$\beta_0$}; 
    \draw [->] (0.1,0.5) -- (4,0.5);
  \end{scope}
  \begin{scope}[yshift=-1cm]
    \draw [->] (0,0) -- (0,1); 
    \draw [->] (0,0) -- (4,0); 
    \node at (-0.5,0.5) {$\beta_2$}; 
    \draw [->] (3.5,0.5) -- (4,0.5);
  \end{scope}
\end{tikzpicture}}
\caption{\small \label{fig:explicit_example} (a) A sphere with
  constant function everywhere except one hotspot. The partition
  of the space is along the axes. Each partition is then
  thickened slighlty as described in
  Section~\ref{sec:part-comp-homol} to give the covering. (b) The
  nerve of the covering. Note that the faces have tetrahedra as
  all 4 vertices of the square intersect. (c) The persistent
  barcode of the sphere, $H_1$ is trivial so it is not shown.}
\end{figure}

We can divide the computation on the sphere into octants divided by
the axis planes, as in Figure~\ref{fig:explicit_example}(a) and
(b). By growing each octant by the tubular neighbourhood method in
Section \ref{sec:cube-partitions}, we get separate computational
units, with contractible intersections of all higher orders. Hence,
after computing persistent homology on all patches, and on pairwise,
triple and quadruple intersections, we get a double complex of
persistence modules shown in Figure~\ref{fig:example_comp}(b).

\begin{figure}
\subfloat[ $E^1$ page.]{
\begin{tikzpicture}[xscale=1.4]
 \node (E40) at (4,0) {$\kk[t]^6$};
  \node (E30) at (3,0) {$\kk[t]^{24}$};
  \node (E20) at (2,0) {$\kk[t]^{24}$};
  \node (E10) at (1,0) {$\kk[t]^8$};
  \node (E11) at (1,1) {$\kk[t]/t^m$};
  \node (E21) at (2,1) {$0$};
  \node (E31) at (3,1) {$0$};
  \node (E41) at (4,1) {$0$};
  \draw [->] (E20.west) -- (E10.east);
  \draw [->] (E30.west) -- (E20.east);
  \draw [->] (E40.west) -- (E30.east);   
  \draw [->] (0.5,-0.4) -- (0.5,1.3); 
  \draw [->] (0.5,-0.4) -- (4.2,-0.4); 
\end{tikzpicture}}\hfill
\subfloat[ Computing $d^2$ on $E^2$. ]{
\begin{tikzpicture}[xscale=1.4]
  \node (E40) at (4,0) {$0$};
  \node (E30) at (3,0) {$\kk[t]$};
  \node (E20) at (2,0) {$0$};
  \node (E10) at (1,0) {$\kk[t]$};
  \node (E11) at (1,1) {$\kk[t]/t^m$};
  \draw [->] (E30) -- (E11); 
  \draw [->] (0.5,-0.4) -- (0.5,1.3); 
  \draw [->] (0.5,-0.4) -- (4.2,-0.4); 
\end{tikzpicture}}
\hfill
\subfloat[ $E^3$ page, collapsed spectral sequence. ]{
  \begin{tikzpicture}[xscale=1.4]
  \node (E40) at (4,0) {$0$};
  \node (E30) at (3,0) {$t^m\kk[t]$};
  \node (E20) at (2,0) {$0$};
  \node (E10) at (1,0) {$\kk[t]$};
  \node (E11) at (1,1) {$0$}; 
  \draw [->] (0.5,-0.4) -- (0.5,1.3); 
  \draw [->] (0.5,-0.4) -- (4.2,-0.4); 
\end{tikzpicture}}
\caption{\small \label{fig:example_comp} Example computation for the sphere.}
\end{figure}
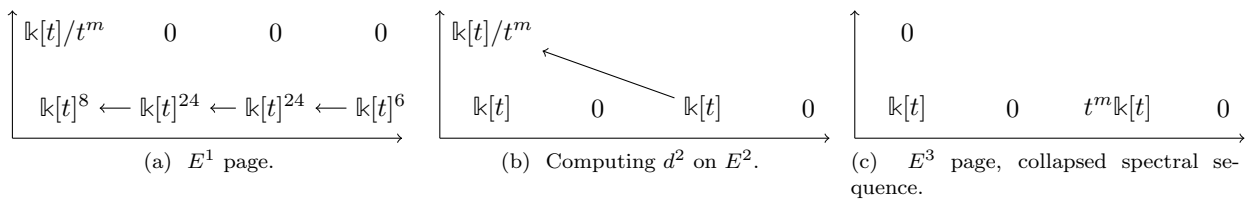

The exponents come from the 6 vertices of the octahedron generating 6
quadruple intersection. In each of these, there are 4 different triple
intersections, generating a total of 24 triple intersections. Each of
the 6 tetrahedra has 6 edges, corresponding to all the double
intersections at each vertex -- however, this overcounts each the 12 double
intersections corresponding to edges of the octahedron. Thus, out of
the 36 double intersections, only 24 remain after compensating for the
overcounting. Finally, each of the 8 faces of the octahedron
corresponds to a single covering patch. 

After computing the corresponding $E^2$-page, the redundancies in
the duplicate representations of all components have largely
resolved, and we are left with a double complex of persistence
modules shown  in Figure~\ref{fig:example_comp}(b).
Here, the map $k[t]\to k[t]/t^m$ maps the generator onto the
generator. The kernel of this map is $t^mk[t]\subset k[t]$, and the
image is the entire module. This gives our final page of the spectral
sequence shown in Figure~\ref{fig:example_comp}(c).
Reading off the modules diagonally, we recover the persistent
homology groups $H_0=k[t]$ and $H_2=t^mk[t]$, which is precisely
what we computed before. 

\section{Applications}
\label{sec:applications}

\subsection{Computing partitions}
\label{sec:part-comp-homol}

The input to use the Mayer-Vietoris spectral sequence is a
filtered simplicial complex is covered by subcomplexes.  Given a
point cloud $\Xcpx$, and a method $C_*$ (such as Vietoris-Rips or
Witness complexes) to construct filtered a simplicial complex
$C_*\Xcpx$ from $\Xcpx$, we shall describe a few schemes to
divide $C_*\Xcpx$ into a covering by appropriate subcomplexes
with small intersections. Minimizing the size of the
intersections, we minimize the amount of work which must be done
in higher pages.

Recall that the \emph{closed star} in a simplicial complex of a
subset of its simplices is the subcomplex consisting of all
simplices that intersect any simplex in the subset. We now state
two lemmas and refer the reader to
Appendix~\ref{sec:proofs-select-result} for the proofs.

\begin{lemma}\label{lemma:closedstartubular}
  Suppose $C_*$ has the property that any simplex has diameter at most
  $\varepsilon/2$, and that the presence of any simplex is decided using
  points of a distance at most $\varepsilon/2$ from any of the vertices
  of the simplex.
 
  Then the closed star of a subset $S\in C_*\Xcpx$ is going to be
  contained in $C_*T_\varepsilon S$ for $T_\varepsilon S$ the $\varepsilon$-tubular
  neighbourhood of $S$ in $\Xcpx$, the collection of points
  $x\in\Xcpx$ such that $d(x,S)<\varepsilon$.
\end{lemma}

\begin{lemma}\label{lemma:subcomplexcovering}
  Suppose $\Xcpx$ is partitioned into subsets $\Ucpx_j$. Then a
  covering of $C_*\Xcpx$ by subcomplexes is given by the collection of
  subcomplexes $C_*T_\varepsilon\Ucpx_j$.
\end{lemma}

We now describe some concrete methods for computing the covers.

\paragraph{\bf Cube partitions}
\label{sec:cube-partitions}

Suppose $\Xcpx$ is a subset of some Euclidean space $\mathbb
R^n$. Cover $\Xcpx$ by a disjoint collection of axis-parallel
cubes.  Membership in a particular cell is quickly computed by
comparing the coordinates of the points in $\Xcpx$ to boundaries
for the cubes.  A covering of $C_*\Xcpx$ is given by the closed
stars of the cells, easily computable as
$C_*T_\varepsilon\Ucpx_j$ for points in the tubular neighborhoods
of the cells.

\paragraph{\bf Voronoi partitions}
\label{sec:voronoi-partitions}

Pick some collection of landmark points $L\subset\Xcpx$. Compute
Voronoi regions $V_\ell$ for $\ell\in L$.  As above, we construct
a covering of $C_*\Xcpx$ by closed stars of the Voronoi cells,
computable as $C_*T_\varepsilon\Ucpx_j$.

\paragraph{\bf Geodesic Voronoi partitions}
\label{sec:geod-voron-part}

Pick some collection of landmark points $L\subset\Xcpx$. Let $C_*$ be
a method that constructs a flag complex on some graph induced by the
geometry of $\Xcpx$. Write $d_\Gamma$ for the geodesic graph distance,
given by $d_\Gamma(x,y)$ equal to the least number of edges from $x$
to $y$ in $\Gamma$. Write $\Ucpx_\ell$ for the set $\{y\in\Xcpx:
d_\Gamma(x,\ell)<d_\Gamma(x,\ell') \forall \ell'\in L\setminus\ell\}$,
or in other words the set of points closer to $\ell$ than to any other
point in $L$.

As a covering by subcomplexes, we grow each $\Ucpx_\ell$ to
$\Ucpx'_\ell$ by including all immediate neighbours of the points in
$\Ucpx_\ell$. A covering of $\Xcpx$ by subcomplexes is given by the
flag complexes on $\Ucpx'_\ell$.

\begin{figure}
  \subfloat{
    \includegraphics[width=3.3cm,page=1]{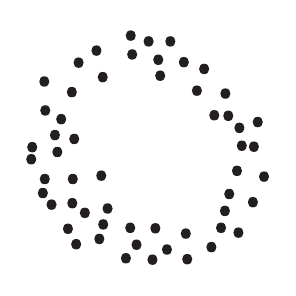}}\qquad
  \subfloat{
    \includegraphics[width=3.3cm,page=2]{voronoi}}\qquad
  \subfloat{
    \includegraphics[width=3.3cm,page=3]{voronoi}}\qquad
  \subfloat{
    \includegraphics[width=3.3cm,page=4]{voronoi}}\qquad
  \caption{\small Regions for Vietoris-Rips in an ambient Voronoi
    decomposition.  From a pointcloud (far left) we pick three
    landmark points, and get a Voronoi decomposition of the
    ambient space (middle left). By growing each Voronoi region
    by $\varepsilon$, we get a covering of the point cloud
    satisfying the conclusion in Lemma
    \ref{lemma:subcomplexcovering} (middle right).  The
    $\varepsilon$-Vietoris-Rips complexes in each of the covering
    patches and their intersections (far right) will provide a
    covering of the $\varepsilon$-Vietoris-Rips complex of the
    entire point cloud by subcomplexes.  }
  \label{fig:vietoris_blowup}
\end{figure}

\subsection{Parallelism and stratified memory}
\label{sec:parall-strat-memory}

Parallelisation with a covering is performed by splitting
responsibility for the chain complex of a single covering patch
or patch intersection to a single computational task. These tasks
communicate with each other to update the global state in the
horizontal differential parts of the computation, and perform
boundary map preimage computations internally. If a computational
node is responsible for chains over a single simplex of the nerve
complex, it will only communicate with nodes representing
simplices in its closure and its star.  This can be seen since in
the first iteration, that node will only need to communicate with
the nodes responsible for its faces and cofaces. In the second
iteration, it must communicate with the faces of its faces and
cofaces of its cofaces. 

To compute the total number of communications produced, we only
consider communication at the level of the nerve of the covering
rather than the size of the underlying simplicial complex. To
avoid double counting, we only consider communication from a
simplex into its closure.

To get a worst-case bound, we assume that the number of
iterations required is equal to the maximal dimension of the
nerve\footnote{As mentioned in
  Section~\ref{sec:spectr-sequ-diff}, the number of iterations is
  bounded by the minimum of the maximal dimensions of the nerve
  and the simplical complex} which we denote with $r$.  Over all
iterations, a $d$-simplex may produce no more than $O(2^d)$
communications. This follows from the assumption that $r>d$ and
the fact that the size of a $d$-simplex's closure is $2^d$.  If
we denote the number of $i$-dimensional simplices in the nerve
with $K_i$ and assume $r$ is the maximal dimension of the nerve,
the total number of communications is $\sum_{i=1}^{r}
K_i 2^i $. This is very rough and pessimistic bound. We intend to
do a more detailed analysis in future work.


We note that this approach will also be applicable in a
single-processor context. Instead of distributing computational tasks
on separate nodes, we can consider each subcomplex or subcomplex
intersection a memory block. This way, the spectral sequence approach
structures a computation, grouping accesses to related memory blocks
into cohesive computational tasks, which induces swap schemes and
schemes for stratified memory access.

\subsection{Nerve lemma}
\label{sec:nerve-lemma}

The Nerve lemma \cite[Corollary 4G.3]{hatcher2005algebraic} is a
powerful result from algebraic topology that has seen repeated use in
computational topology -- all the way from motivating why the commonly
used filtered simplicial complex constructions capture the original
topology \cite{c09TandD} to generating particular new results.

Quite often, when computational topologists use the Nerve lemma,
quite a lot of effort is expended to provide the contractibility
conditions required in the classical statement of the Nerve
lemma. This is due to the most common form of the Nerve lemma
guarantees both homological and homotopical equivalence between a
covering and its nerve. As we are usually only interested in
homology, we observe that the Mayer-Vietoris spectral sequence
yields a proof of a homological Nerve lemma:

\begin{theorem}[Homology Nerve lemma]\label{thm:nervehomology}
  Suppose a space $\Xcpx$ is covered by a collection $\Ucoll$ of
  subspaces $\Xcpx = \bigcup_i\Ucpx_i$, such that each
  $\bigcap_j\Ucpx_j$ has trivial reduced homology in all degrees. Then
  $H_*(\Xcpx) = H_*(\mathcal N(\Ucoll))$.
\end{theorem}

\begin{proof}
  On the $E^1$-page, the  non-reduced homology is concentrated in
  degree 0. Since by assumption, the reduced homology is trivial
  for each intersection, the non-reduced homology is just a copy
  of the ground field. The differential on the $E^1$-page is the
  differential of the nerve complex itself, and the result
  follows immediately.
\end{proof}

While this is a well-known fact in classical algebraic topology,
we point this out to emphasize that the intersections need
not be equipped with explicit contracting homotopies for the
Nerve Lemma to work in homological computations.

\paragraph{\bf Persistent Nerve Lemma:} The Nerve Lemma was extended
to the persistence setting in~\cite{co-tpbr-08} using similiar homotopical
arguments. Here, we give the spectral sequence version of it
 which relies on a persistent acyclicity condition.  

From our framework also follows a persistent Nerve lemma.
\begin{definition}
  A space is \emph{persistently acyclic} if it has persistent homology
  consisting of one free generator of $H_0$, and no homology in higher degrees.
\end{definition}

\begin{theorem}[Persistent homology Nerve lemma]
  Suppose a filtered space $\Xcpx$ is covered by a collection $\Ucoll$
  of filtered subspaces $\Xcpx = \bigcup_i\Ucpx_i$, such that each
  $\bigcap_j\Ucpx_j$ is persistently acyclic. Then $H_*(\Xcpx) =
  H_*(\mathcal N(\Ucoll))$.
\end{theorem}

\begin{proof}
  We note that persistent homology is homology with coefficients
  in $\kk[t]$. The persistent acyclicity condition ensures that
  homology has just one generator with the appropriate
  grading concentrated in the 0-homology. As in Theorem
  \ref{thm:nervehomology}, the sequence collapses after the $E^1$
  page and the $d^1$ differential is equivalent to the boundary
  operator of the nerve. The result follows.
\end{proof}


\subsection{Algebraic kernel, cokernel and image persistence}
\label{sec:algebr-kern-cokern}

In~\cite{cohen2009persistent}, the authors construct topological
spaces that compute \emph{image}, \emph{kernel} and \emph{cokernel}
persistence modules for the map induced in persistent homology by a
map between filtered topological spaces. We observe that the
constructions in section \ref{sec:algebraic-operations} provide
corresponding operations on a purely algebraic level, computing
persistence modules for \emph{image}, \emph{kernel} and
\emph{cokernel} given a map of persistence modules, such as the one
induced from a topological map.

The notion of compatible filtration included as a condition
in~\cite{cohen2009persistent} is reflected in the requirement that
maps between persistence modules be \emph{graded} module homomorphisms
of degree 0, implying that the filtration inclusions on both sides
commute with the map itself.

\section{Discussion}
\label{sec:discussion}

We have constructed algorithms to compute with the spectral sequence
generated by a double complex. In particular, we describe how to
compute a spectral sequence generalizing the Mayer-Vietoris long exact
sequence to compute both classical and persistent homology using a
cover of the space and restricting computation to each of the covered
regions and their intersections.

This layout enables us to formulate parallelization schemes, with
parallelism and independency guarantees present both along the Nerve
complex axis of our computational scheme and along separate components
of the covering.

As an added bonus, the formalism of the spectral sequence
approach suggests simplified proofs for several classical
interesting results in the persistence community, and opens up a
wide spectrum of interesting areas of research. In particular, we
are interested in pursuing:
\begin{description}
\item[Complexity bounds for parallelism] This paper presents the
  formalism, but does not attempt a complete description of the
  complexity behaviour of the proposed algorithms.
\item[HPC implementation] While the work was driven by a wish to take
  persistent homology to the world of computing clusters, our
  implementation so far is a serial one. A parallel implementation,
  and actual performance measurements would be very interesting.
\item[Cover generation schemes] While we give a few suggestions here,
  the construction of an adequate cover that uses or highlights
  inherent features of the dataset under analysis is going to be an
  important factor of these methods. A more complete catalogue of
  cover generation schemes is going to be important for the practical
  application of our techniques.
\item[Theoretical implications] We have already seen that several
  results have algebraic proofs that vastly simplify the arguments
  involved. We would be very interested in exploring the space of new
  or classical results that become accessible with a spectral sequence
  language.
\end{description}

\clearpage
\bibliographystyle{abbrv}
\bibliography{mv}

\clearpage
\appendix

\section{Proofs of selected results}
\label{sec:proofs-select-result}

\begin{proof}[Proof of Lemma \ref{lemma:closedstartubular}]
  Suppose $C_*$ determines the presence of a simplex using points at most
  $\varepsilon/2$ away from any point in a simplex, and each simplex has
  diameter at most $\varepsilon/2$. A simplex $\sigma$ in $C_*\Xcpx$ is in the
  closed star of $S$ if at least one vertex of $\sigma$ is in
  $S$. By the condition on the diameter, all other points of $\sigma$
  are contained in $T_{\varepsilon/2} S$. In addition, any point that
  can influence the decision on whether $\sigma$ exists or not will be
  within $\varepsilon/2$ of the furthest point from the vertex of
  $\sigma$ in $S$, and thus the simplex $\sigma$ is entirely
  determined by points at most $\varepsilon$ away from $S$, and thus
  by $T_{\varepsilon} S$.
\end{proof}

\begin{proof}[Proof of Lemma \ref{lemma:subcomplexcovering}]
  If $\Xcpx$ is partitioned into subsets $\Ucpx_j$, we may certainly
  partition $C_*\Xcpx$ by the closed stars of the subsets. Indeed,
  since the $\Ucpx_j$ partition $\Xcpx$, we are guaranteed that any
  vertex of a simplex lies in at least one of the $\Ucpx_j$. Hence, in
  particular, the closed stars of all the $\Ucpx_j$ will contain all
  simplices in $C_*\Xcpx$.

  By Lemma \ref{lemma:closedstartubular}, each such closed star is
  contained in $C_*T_\varepsilon\Ucpx_j$. Hence, the subcomplexes
  $C_*T_\varepsilon\Ucpx_j$ cover $\Xcpx$.
\end{proof}

\section{Implementation}\label{sec:implementation}

\begin{figure} 
\centering\includegraphics[width=0.75\textwidth]{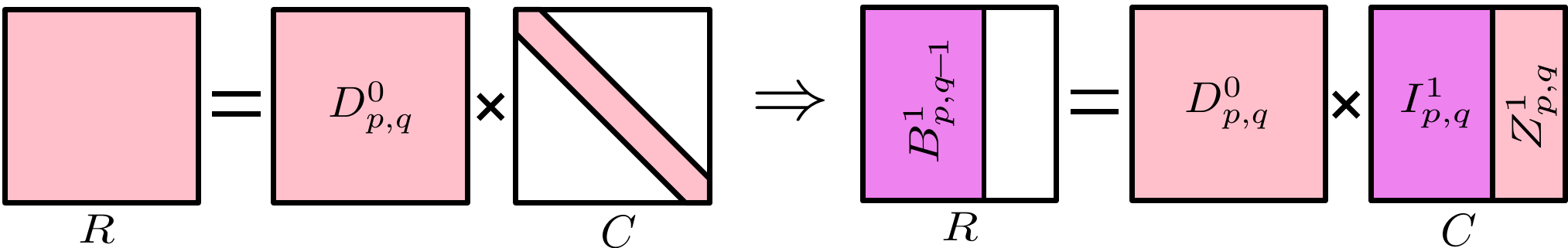}
\caption{\small \label{fig:basic_reduce}\small The basic reduction: As
  in~\cite{cem-vvuplt-06}, we reduce the matrix $R$ left to right, always
  choosing the lowest possible pivot with shaded areas
  representing non-zero entries. In contrast with~\cite{cem-vvuplt-06}, we
  represent each dimension separately, so the matrices are not upper
  triangular. As we reduce $R$, we perform the same operations on
  $C$. At the end, we perform a permutation to put the columns
  with pivots on the left (this is not done in practice). These
  represent the boundary basis and the corresponding columns in
  $C$ are the preboundary basis. The columns in $C$ corresponding
  to the empty columns in $R$ are the cycle basis.}
\end{figure}
In this section, we recount the algorithm in
Section~\ref{sec:spectr-sequ-diff} in greater detail. In
particular, the sections follow each other closely. 

 The basic element of computation is the chain which is
 represented by a column vector with each row corresponding to a
 simplex and the entries corresponding to a coefficient
 $\kk$. Since the chains are homogeneous, it is sufficient to
 annotate each chain with its degree. Bases are a collection of
 linearly independent chains, and are represented as a matrix,
 where the vectors are sorted from left to right in order of
 increasing degree\footnote{This corresponds to an increasing
   filtration index.}.

A persistent homology class is represented as a cycle chain and a
boundary chain which represent the birth and death times of the
class respectively. In terms of the presentation of the module,
an all classes have a representative in the cycle basis where the
degree corresponds to birth time. Inessential classes have a copy
of this with higher degree in the boundary basis. In the quotient
space, this corresponds precisely to the algebraic form described
in~\cite{cohen2009persistent}. In practice, we do not choose
quite this representation.  Rather than keep a separate boundary
basis and cycle basis, inessential classes are annotated with two
degrees to represent the boundary and cycle basis degrees.

\begin{figure} 
\centering\includegraphics[width=0.9\textwidth]{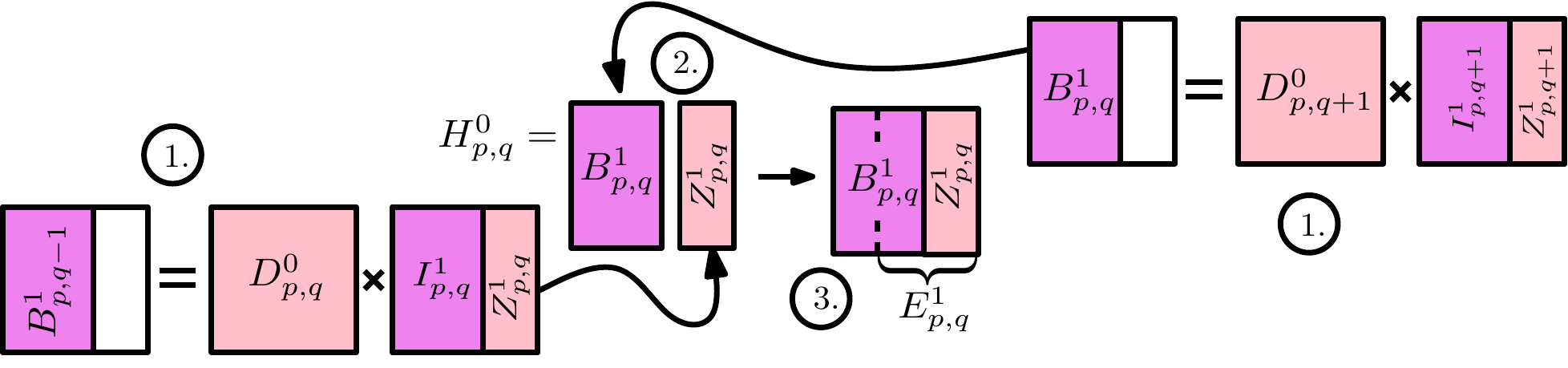}
\caption{\small \label{fig:boundary_reduce} The reduction procedure:
  1. The standard reduction into cycles and boundaries
  2. Reduction of cycle basis with respect to cycle basis, the
  reduction proceeds left to right. Note that certain chains
  (non-essential classes) in $Z^1_{pq}$ will be completely
reduced. Since these are already in the span of the chosen basis
we  drop those columns. 3. The input to the next iteration
ignores null persistent classes (in the picture we assume they
happen in a left block, but this is only for visualization).}
\end{figure}

\paragraph{\bf Initialization}: The initialization corresponds to 
running the standard persistence algorithm on each column in the
double complex with $d^0$.  As in~\cite{cem-vvuplt-06}, by
reducing the boundary matrix and keeping track of column
operations we can extract all the required bases at each node
$(p,q)$(Figure~\ref{fig:basic_reduce}). We now perform one
additional operation: at each node, we reduce the cycle basis
with respect to the boundary basis as shown in
Figure~\ref{fig:boundary_reduce}. We do this additional step to
make operations at later stages simpler.

To keep track of the preboundary we track the operations
performed to produce the non-zero reduced columns of the boundary
matrix. This way, we get a minimal spanning set of $I^1_{p,q}$
such that $d^0$ yields an isomorphism $I^1_{p,q}\to B^1_{p,q-1}$.
At the end of this step, we have matrices representing
$Z^1_{p,q}$, $B^1_{p,q}$ and $I^1_{p,q}$ for all $(p,q)$.
$Z^1_{p,q}$ and part of $B^1_{p,q}$ together form a basis for
$E^1_{p,q}$. We omit the boundary elements because $E^r$ is a
subquotient of all earlier $E^j$, so any non-trivial class in
$E^r$ is represented earlier, and any relation from earlier
remains a relation. Hence, any null-persistent class may show up
as a relation, but only a relation killing its corresponding
cycle, which also will no longer be relevant. 
Therefore they are still valid relations for the basis but cannot
map to any element in $E^r_{p-r+1,q+r}$.

\paragraph{\bf Iteration}: The general iteration computes the basis 
for $Z^r$, $B^r$, $I^r$, from which we compute $E^{r+1}$. Note
that these do not need to be stored over multiple iterations. We
store the 0-th iteration from the initialization and then only
the previous iteration: For the $r$-th iteration, we only have
$(r-1)$-th iteration. Once completed we can forget the $(r-1)$-th
iteration. As before, we first assume that all the
$r$-differentials are given as linear maps. Practically, we
compute the differential from the information in each subsequent
iteration. This can be thought of as a subroutine we call at the
beginning of each iteration and is described below.

The rest of the iteration follows much as the initialization, we
apply $d^r$ to $E_{p,q}^{r}$. Each element gives a chain $x \in
E_{p-r+1,q+r}^{r}$. We first reduce $x$ with respect to the
existing boundary chains. In Figure~\ref{fig:one_step}, we show how
we perform a reduction step. We reduce the chain $x$ with
respect to $y$. The underlying operation is the same as for
regular reduction in Gaussian elimination: we multiply the pivot
column with the appropriate field coefficient and add it to the
chain we are reducing. The only additional book-keeping we must
do is to take care that the degrees of the cycle and boundary
chains are appropriately updated. 

Assume that both $x$ and $y$ are inessential classes, so each has a
generator and a boundary relation. To illustrate this,
Figure~\ref{fig:one_step} shows the mapping between two persistence
interval, representing the degrees of the generators and relations.
Such a mapping would generically map some interval $[c,d]$ to an
interval $[a,b]$, with $a \leq c\leq b\leq d$.  This corresponds to
the requirement that the maps be graded so generators map to multiples
of generators and relations map to relations. The kernel of the map is
represented by $x$ shifted to the degree of the relation of $y$. The
cokernel is represented by $y$ with the relation at the degree of the
generator of $x$. The image is represented by $y$, with a degree shift
up to the generator of $x$.  The transformation of the intervals is
shown in Figure~\ref{fig:one_step}.

This corresponds to only one reduction step, we repeat this
procedure completely reducing $x$ in terms of the basis in
$E_{p-r+1,q+r}^{r}$ and we do this for all elements in
$E_{p,q}^{r}$. If in the reduction of an element, an interval is
reduced to zero (i.e. $[c,c]$), further reduction is unnecessary,
since it cannot be in the kernel, nor can can it map to anything
further in the boundary.  The kernel of the map is given by all
the elements which have non-null persistence and form a basis for
$Z_{p,q}^{r+1}$. Likewise, in the target space (assuming
$Z_{p-r+1,q+r}^{r+1}$ had already been computed), the remaining
non-null persistent classes form a basis
$E_{p-r+1,q+r}^{r+1}$. We note that we can compute the
basis for the boundary and kernel independently. However, then we
must reduce the kernel basis with repect to the boundary basis
(just as in the initialization).


\begin{figure} 
\centering\includegraphics[width=0.7\textwidth]{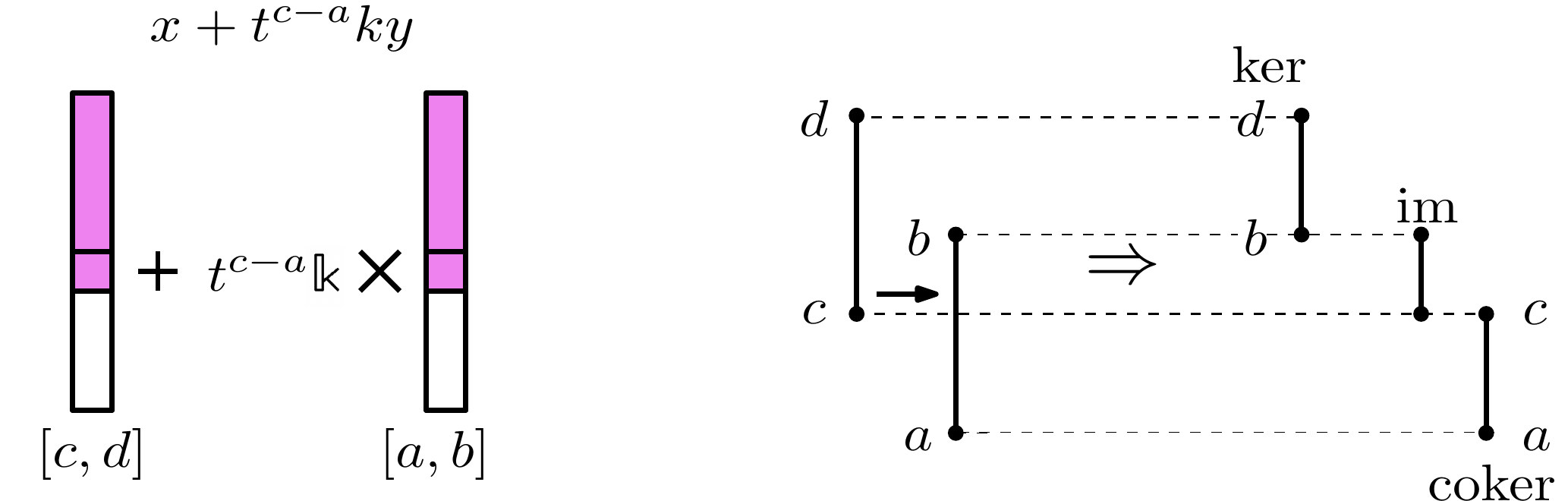}
\caption{\small \label{fig:one_step} To perform a reduction of
  the $x \in d^r E^r_{p+r-1,q-r}$ with an element of the basis of
  $y \in E^r_{p,q}$, we multiply by the appropriate coefficient
  $t^{c-a}k\in\kk[t]$ and add the two. The condition $a\leq c\leq b
  \leq d$ follows from the grading of the maps. We update the
  degrees of $d^r E^r_{p+r-1,q-r}$ to $[b,d]$ and $E^r_{p,q}$ to
  $[a,c]$.  Note that if we get $[d,d]$ or $[a,a]$, these
  elements have become null persistent.  }
\end{figure}

\paragraph{\bf Computing the $r$-differential}: For
$r=1$, the differential is given by $d^1$, where the map on homology
is induced directly from the chain map.  

For $r>1$, we use the following procedure to compute $d^r$ shown
in Figure~\ref{fig:lift}. In the previous iteration, we store the
image of $d^{r-1} E^{r-1}_{p,q}$. More precisely, we store the
representation of the image in the basis of
$B^0_{p-r+2,q+r-1}$. In this basis, performing the lift to
$I^0_{p-r+2,q+r}$ can be done through the stored relation.
%
%
%
The lift exists because from
Section~\ref{sec:mayer-viet-spectr}, we know that $d^{r-1}
E^{r}_{p,q} \in B^0_{p-r+2,q+r-1}$.

The last step is applying the $d^1$ map, which gives the chain we then
reduce with respect to the $E^{r}_{p-r+1,q+r}$ basis. By
construction this satisfies equations~\ref{eq:cycle_equations}
and ~\ref{eq:boundary_equations}. This is precisely one step of
the staircase shown in Figure~\ref{fig:cyclesandboundaries}.

\begin{figure} 
\centering\includegraphics[width=0.7\textwidth]{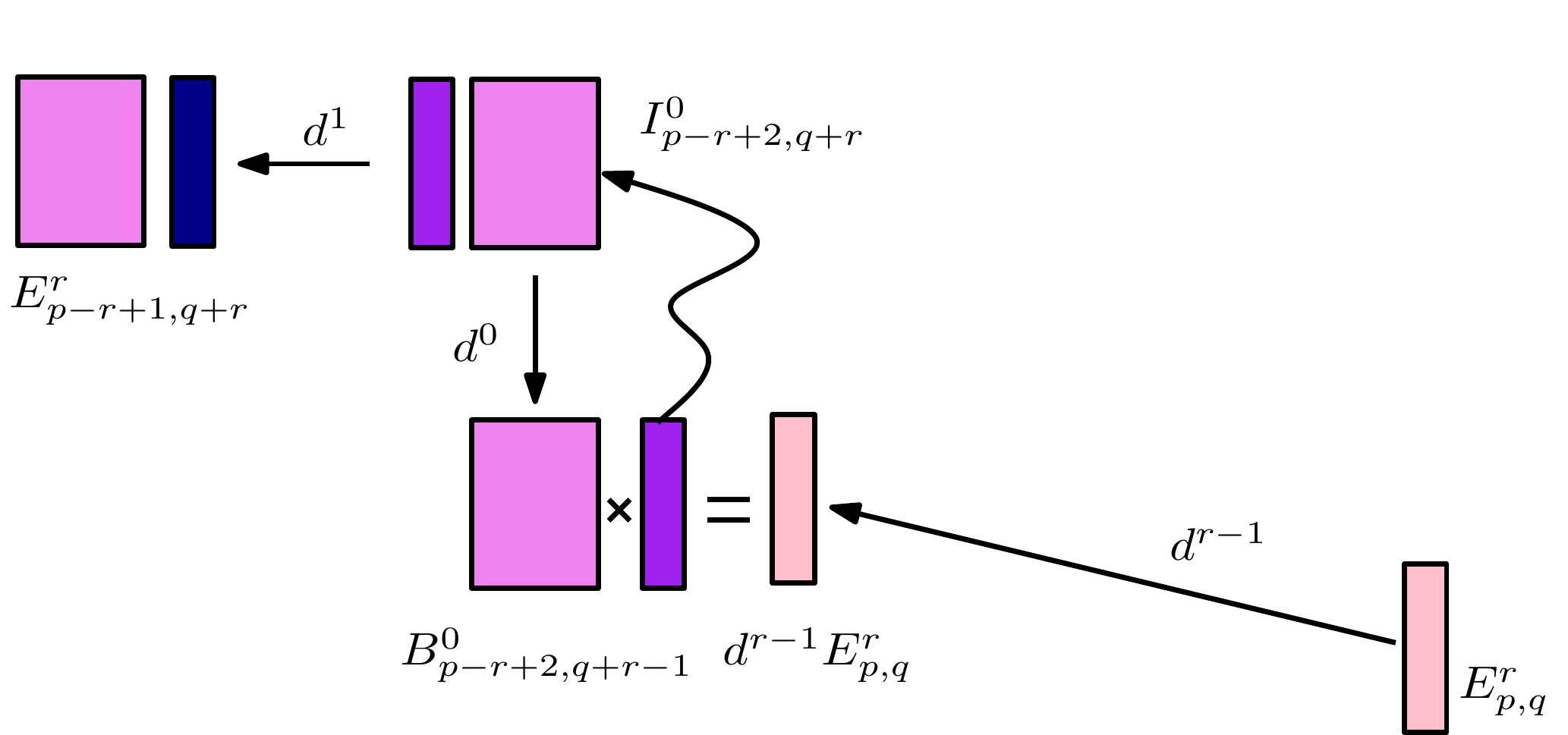}
\caption{\small \label{fig:lift} To compute $d^r$, we first find the
  representation of the $d^{r-1} E^r_{p,q}$ in terms of
  $B^0_{p-r+1,q+r-1}$. Since we know this is in the image of
  $d^0$, we can lift this vector to $I^0_{p-r+1,q+r}$ and then
  compute the chain in terms of $I^0_{p-r+1,q+r}$. Applying
  $d^1$, we get a chain we can reduce in terms of $E^r_{p-r+1,q+r}$, thereby computing $d^r$}
\end{figure}

\end{document}